\documentclass[a4paper,USenglish,cleveref,numberwithinsect,thm-restate]{lipics-v2021}

\usepackage{microtype} 
\usepackage{todonotes} 
\usepackage{cleveref} 

\graphicspath{{./figs/}}

\usepackage{xspace}
\usepackage{xcolor}
\usepackage{array,multirow}
\usepackage{arydshln}

\newcommand{\AWNN}{\textsf{AWNN}\xspace}
\newcommand{\lca}{\textsf{lca}\xspace}

\newcommand{\R}{\mathbb{R}}

\newcommand{\cD}{\mathcal{D}}

\title{Insertion-Only Dynamic Connectivity in General Disk Graphs}

\author{Haim Kaplan}{School of Computer Science, Tel Aviv University}
{haimk@tau.ac.il}{}{}
\author{Katharina Klost}{Institut f\"ur Informatik, Freie Universität Berlin}
{kathklost@inf.fu-berlin.de}{}{}
\author{Kristin Knorr}
{Institut f\"ur Informatik, Freie Universität Berlin}{knorrkri@inf.fu-berlin.de}
{}{Supported
by the German Science Foundation within the research
training group `Facets of Complexity' (GRK 2434).}
\author{Wolfgang Mulzer}
{Institut f\"ur Informatik, Freie Universität Berlin}
{mulzer@inf.fu-berlin.de}{https://orcid.org/0000-0002-1948-5840}
{Supported in part by ERC StG 757609.}
\author{Liam Roditty}{Department of Computer Science, Bar Ilan University}
{liamr@macs.biu.ac.il}{}{}
\authorrunning{H. Kaplan, K. Klost, K. Knorr, W. Mulzer, and L. Roditty}

\hideLIPIcs
\nolinenumbers

\ccsdesc[300]{Theory of computation~Computational geometry}

\keywords{disk graphs, dynamic, insertion only}

\begin{document}

\maketitle

\begin{abstract}
Let $S \subseteq \R^2$ be a set of $n$ \emph{sites} 
in the plane, so that every site $s \in S$ has an 
\emph{associated radius} $r_s > 0$. Let $\cD(S)$ be
the \emph{disk intersection graph} defined 
by $S$, i.e., the graph with vertex set $S$ and 
an edge between two distinct sites $s, t \in S$ 
if and only if the disks with centers $s$, $t$ 
and radii $r_s$, $r_t$ intersect.
Our goal is to design data structures that 
maintain the connectivity structure of $\cD(S)$ 
as $S$ changes dynamically over time.

We consider the incremental case, where new sites 
can be inserted into $S$. While previous work focuses on data 
structures whose running time depends on the ratio between the 
smallest and the largest site in $S$, we present a data 
structure with $O(\alpha(n))$ amortized query time and $O(\log^6 n)$ 
expected amortized insertion time. 
\end{abstract}

\section{Introduction}

The question if two vertices in a given graph are connected 
is crucial for many applications.
If multiple such \emph{connectivity queries} need to be answered,
it makes sense to preprocess the graph into a suitable data
structure. In the static case, where the graph does not
change, we get an optimal answer by using a graph search 
to determine the connected components and by labeling the 
vertices with their respective components.
In the dynamic case, where the graph can change over time,
things get more interesting, and many variants of the problem
have been studied.

We construct an insertion-only 
dynamic connectivity data structure for disk graphs. 
Given a set $S \subseteq \R^2$ of $n$ sites in the plane 
with associated radii $r_s$ for each site $s$, the \emph{disk graph  $\cD(S)$ for $S$} 
is the intersection graph of the disks $D_s$ induced by the sites and 
their radii. While $\cD(S)$ is represented by
$O(n)$ numbers describing the disks, it might have $\Theta(n^2)$ edges. 
Thus,  when we start with an empty disk graph and successively 
insert sites, up to $\Omega(n^2)$ edges may be created.
We describe a data structure whose overall running time for 
any sequence of $n$ site insertions is $o(n^2)$, while
allowing for efficient connectivity queries.
For unit disk graphs (i.e., all associated radii $r_s = 1$), 
a fully dynamic data structure with a similar performance guarantee
is already given by Kaplan et al.~\cite{kaplan_dynamic_2022a}.
In the same paper, Kaplan et al.\@ present an incremental data structure 
whose running time depends on the ratio $\Psi$ of the smallest and the 
largest radius in $S$. In this setting, they achieve $O(\alpha(n))$ 
amortized query time and $O(\log(\Psi)\log^4 n)$ expected amortized 
insertion time.

We focus on general disk graphs, with no assumption on the  
radius ratio.
Our approach has two main ingredients. 
First, we  simply represent the
connected components in a data structure for the disjoint set union problem \cite{cormen_introduction_2009}.
This allows for fast queries, in $O(\alpha(n))$ amortized time, also mentioned by Reif \cite{reif_topological_1987}.
The second ingredient is an efficient data structure
to find
all components in $\cD(S)$ that are intersected by any given disk (and hence
tells us which components need to be merged after an insertion of
a new site).
A schematic overview of the data structure is given in \autoref{fig:ins_unbounded}.

\begin{figure}
    \centering
    \includegraphics{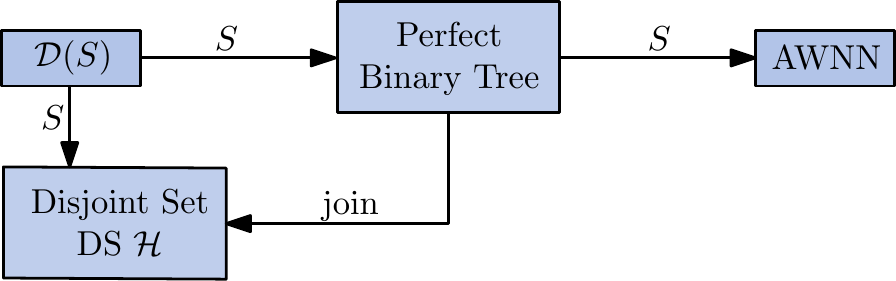}
    \caption{Overview of the data structure (\autoref{thm:ins_unbounded})}
    \label{fig:ins_unbounded}
\end{figure}

\section{Insertion-Only Data Structure for Unbounded Radii}

As in the incremental data structure for disk graphs 
with bounded radius ratio by Kaplan et al.~\cite{kaplan_dynamic_2022a}, 
we use a disjoint set union data structure to represent the
connected components of $\cD(S)$ and to perform the connectivity queries. 
To insert a new site $s$ into $S$, we first find the set 
$\mathcal{C}_s$ of
components in $\cD(S)$ that are intersected by $D_s$.
Then, once $\mathcal{C}_s$ is known, we can simply update the 
disjoint set union structure to support further queries.

   \begin{figure}
    \centering
    \includegraphics[page=5,width=\linewidth]{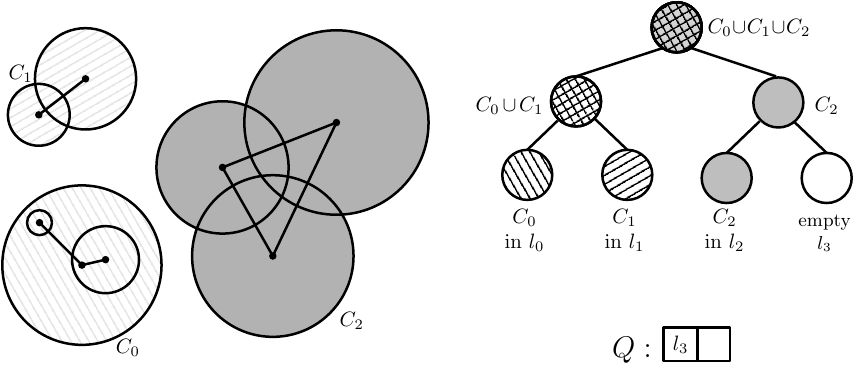}
    \caption{Disk graph with associated component tree and queue $Q$ (tiling of a component corresponds to the tiling of nodes storing its disks).}
    \label{fig:example}
\end{figure}
 
In order to identify $\mathcal{C}_s$ efficiently, 
we use a \emph{component tree} 
$T_\mathcal{C}$.
This is  a binary tree whose leaves store the connected
components of $\cD(S)$.
The idea is illustrated in \autoref{fig:example} showing a disk graph and its associated component tree.
We require that $T_\mathcal{C}$ is 
a complete binary tree, and some of its leaves may not have 
a connected component assigned to them, like $l_3$ in \autoref{fig:example}. Those leaves are \emph{empty}.
Typically, we will not distinguish the leaf storing a connected 
component and the component itself. Also when suitable, 
we will treat a connected component as a set of sites.
Every  node of $T_\mathcal{C}$ stores 
a fully dynamic additively weighted nearest neighbor data structure 
(\AWNN). 
An \AWNN stores a set \(P\) of \(n\) points, each associated with a weight \(w_p\).
On a nearest neighbor query with a point \(q \in \R^2\) it returns the point \(p\in P\) that minimizes \(\Vert pq\Vert + w_p\).
For this data structure, we use the following result by 
Kaplan et al.~\cite{kaplan_dynamic_2020} with an improvement by 
Liu~\cite{Liu20}.
\begin{lemma}[Kaplan et al.~{\cite[Theorem 8.3, Section~9]{kaplan_dynamic_2020}, Liu~\cite[Corollary~4.3]{Liu20}}]\label{lem:prelims:dynamicNN}
There is a fully dynamic \AWNN{} data structure that 
allows insertions in $O(\log^2 n)$ amortized 
expected time and deletions in $O(\log^4 n)$
amortized expected time. 
Furthermore, a nearest neighbor query takes $O(\log^2 n)$ worst case time.
The data structure requires $O(n \log n)$ space.
\end{lemma}
The component tree maintains the following invariants.
Invariant 2 allows us to use a query to the \AWNN to find the disk whose boundary is closed to a query point.
\begin{description}
\item[Invariant 1:] Every connected component of $\cD(S)$ 
is stored in exactly one leaf of $T_\mathcal{C}$, and
\item[Invariant 2:] The \AWNN of a node $u \in T_\mathcal{C}$ contains
the sites of all connected components that lie in the subtree rooted at $u$, 
where a site $s \in S$ has assigned weight $-r_s$.
\end{description}
In addition to $T_\mathcal{C}$, we store a queue $Q_\mathcal{C}$ that 
contains exactly the empty leaves in $T_\mathcal{C}$.

We now describe how to update $T_\mathcal{C}$ when a new site $s$
is inserted. We maintain $T_\mathcal{C}$ in such a way that
the structure of $T_\mathcal{C}$ changes only when $s$ creates
a new isolated component in $\cD(S)$.
In the following lemma, we consider the slightly more general case 
of inserting a new connected component $C$ that 
does not intersect any connected component already stored in~$T_\mathcal{C}$.

\begin{lemma}\label{lm:insertComp}
    Let $T_\mathcal{C}$ be a component tree of height $h$ 
    with $n$ sites and let $C$ be an isolated connected component. 
    We can insert $C$ into $T_\mathcal{C}$ in 
    amortized time $O(h \cdot  |C| \cdot \log^2 n)$.
\end{lemma}

\begin{proof}
The insertion performs two basic steps: first, we find or create an empty leaf 
$l_i$ into which $C$ can be inserted. Second, the \AWNN structures along
the path from $l_i$ to the root of $T_\mathcal{C}$ are updated.

For the first step, we check if $Q_\mathcal{C}$ is non-empty.
If so, we extract the first element from 
$Q_\mathcal{C}$  to obtain our empty leaf $l_i$. 
If $Q_\mathcal{C}$ is empty, there are no empty leaves, 
and we have to expand the component tree.
For this, we create a new root for $T_\mathcal{C}$, and we 
attach the old tree as one child. The other child is an empty complete 
tree of the same size as the old tree. This creates a complete 
binary tree, see intemediate state in \autoref{fig:doubletree}.
We copy the \AWNN of the former root to the new root,  we 
add all new empty leaves to $Q_\mathcal{C}$, and we extract $l_i$
from $Q_\mathcal{C}$.

   \begin{figure}
   \centering
   \includegraphics[page=2,width=\linewidth]{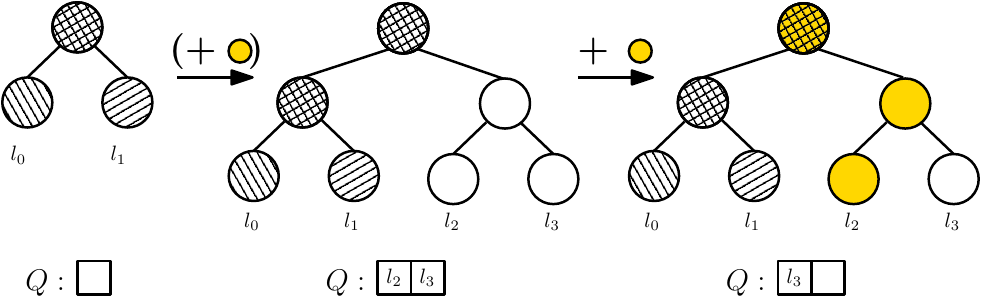}
   \caption{If the tree has no empty leaves before the insertion of an 
   isolated component, a new root and an empty subtree are added (second tree is an intermediate state before actual insertion).}
   \label{fig:doubletree}
   \end{figure}
    
For the second step, we insert $C$  into $l_i$, and we store an \AWNN structure 
with the sites from $C$ in $l_i$.
Then, the \AWNN structures on all ancestors of $l_i$ are updated 
by inserting the sites of $C$, see \autoref{fig:isolatedcomp} and \autoref{fig:doubletree} in case of tree extension respectively.

   \begin{figure}
   \centering
   \includegraphics[page=1]{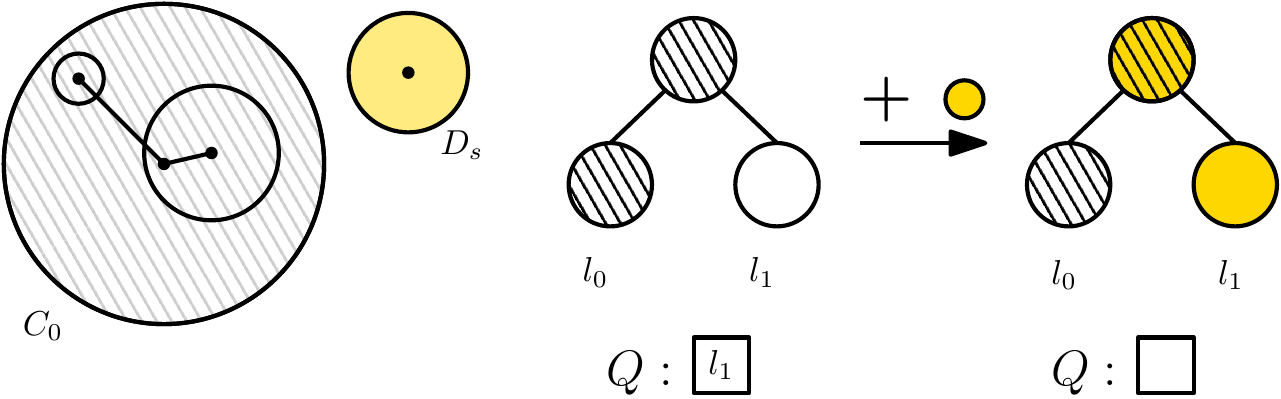}
   \caption{Inserting a component (potentially isolated disk $D_s$) into an empty leaf ($C_i$ stored in $l_i$): The isolated component $C$, yellow disk $D_s$, is inserted in the empty leaf $l_1$ and all its ancestors (indicated by coloring).}
   \label{fig:isolatedcomp}
   \end{figure}
    
    This procedure maintains both invariants:
    since an isolated component does not affect 
    the remaining connected components of $\cD(S)$, 
    it has to be inserted into a new leaf, maintaining the first invariant.
    The second invariant is taken care of in the second step, by
    construction.
    Afterwards,  the queue has the correct state, since we extract 
    the leaf used in the insertion.
    
    The running time for finding or creating 
    an empty leaf is amortized $O(1)$. 
    This is immediate if $ Q_\mathcal{C} \neq \emptyset$, 
    and otherwise, 
    we can charge the cost of building the empty tree, inserting 
    the empty leaves into $Q_\mathcal{C}$, and producing an \AWNN
    structure for the new root to the previous insertions.
    The most expensive step consists in updating the \AWNN structures
    for the new component.
    In each of the $h$  \AWNN structures of the ancestors of $l_i$, 
    we must insert $|C|$ disks. By \autoref{lem:prelims:dynamicNN},
    this results into an expected amortized time of  
    $O(h \cdot  |C| \cdot \log^2 n)$.
\end{proof}

Next, we describe how to find the set $\mathcal{C}_s$ 
of connected components that are intersected by a disk $D_s$.

\begin{lemma}\label{lm:find}
    Let $T_\mathcal{C}$ be a component tree of height $h$ that 
    stores $n$ disks. We can find $\mathcal{C}_s$ 
    in worst case time $O(\max\{|\mathcal{C}_s| \cdot h, 1\}\cdot\log^2 n\})$.
\end{lemma}

\begin{proof}
    First, observe that if the site returned by a query to an \AWNN 
    structure with $s$ does not intersect $D_s$,  then $D_s$ 
    does not intersect any disk for the sites stored in this \AWNN.
    Thus, the case where $\mathcal{C}_s = \emptyset$ 
    can be identified by a query to the \AWNN structure 
    in the root of $T_\mathcal{C}$,  in $O(\log^2 n)$ time. 

    In any other case, we perform a top down 
    traversal of $T_\mathcal{C}$. Let $u$ be the current node. 
    We query the \AWNN structures of both children of $u$ with $s$, 
    and we recurse only into the children 
    where the nearest neighbor intersects $D_s$.
    The set $\mathcal{C}_s$ then contains exactly the connected components 
    of all leaves where $D_s$ intersects its weighted nearest neighbor.
    Since every leaf found corresponds to one connected component 
    intersecting $D_s$ and we recurse into all subtrees whose union 
    of sites have a non-empty intersection with $s$, 
    we do not miss any connected components.

    For every connected component, there are at most $h$ 
    queries to \AWNN structures along the path from the root 
    to the components. A query takes $O(\log^2n)$ amortized time by 
    \autoref{lem:prelims:dynamicNN}, giving an amortized time of 
    $O(|\mathcal{C}_s| \cdot h \cdot \log^2 n)$, if $s$ is not isolated. 
    The overall time follows from taking the maximum of both cases.
\end{proof}

Using \autoref{lm:insertComp} and \autoref{lm:find}, we can now describe how to insert a single disk into a component tree.

\begin{lemma}\label{lm:insertDisk}
Let $T_\mathcal{C}$ be a component tree that stores $n$ sites. 
A new site $s$ can be inserted into $T_\mathcal{C}$ in $O(\log^6 n)$ 
amortized expected time.
\end{lemma}

\begin{proof}
     First, we use the algorithm from  \autoref{lm:find}.
     to find $\mathcal{C}_s$.
    If $|\mathcal{C}_s| = 0$, we use \autoref{lm:insertComp} to 
    insert $s$ as a singleton isolated connected component.
    
     \begin{figure}
    \centering
    \includegraphics[page=3]{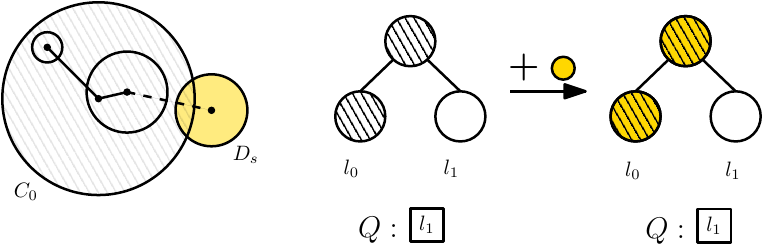}
    \caption{Inserting a site if $|\mathcal{C}_s| = 1$ ($C_i$ stored in $l_i$): The yellow disk $D_s$ intersects $C_0$ (dashed edge). Site $s$ is added to the \AWNN of the associated leaf $l_0$ and its ancestors.}
    \label{fig:insnocleanup}
    \end{figure}
    Otherwise, if $|\mathcal{C}_s| \geq 1$, 
    let $C_L =\arg\max_{C\in\mathcal{C}_s} |C|$ 
    be a largest connected component in $\mathcal{C}_s$. 
    We insert $s$ into $C_L$ and into the \AWNN structures 
    of all ancestors of $C_L$.
    Then, if $|\mathcal{C}_s| = 1$, we are done, see \autoref{fig:insnocleanup}.
    If $|\mathcal{C}_s| > 1$, all components in $\mathcal{C}_s$ 
    now form a new, larger, component in $\cD(S)$. We perform the 
    following \emph{clean-up} step in $T_\mathcal{C}$.

    For each component $C_i \in \mathcal{C}_s \setminus \{C_L\}$, 
    all sites from $C_i$ are inserted into $C_L$.
    Let $\lca$ be the lowest common ancestor of the leaves for $C_i$ and $C_L$ 
    in $T_\mathcal{C}$.
    Then all sites from $C_i$ are deleted from the \AWNN structures 
    along the path from $C_i$ to $\lca$, and  reinserted along the path 
    from  $C_L$ to $\lca$.
    Finally, all newly empty leaves are inserted into $Q_\mathcal{C}$.
    For an illustration of the insertion of $s$ 
    and the clean-up step, see \autoref{fig:inscleanup}.

    \begin{figure}
    \centering
    \includegraphics[page=4,width=\linewidth]{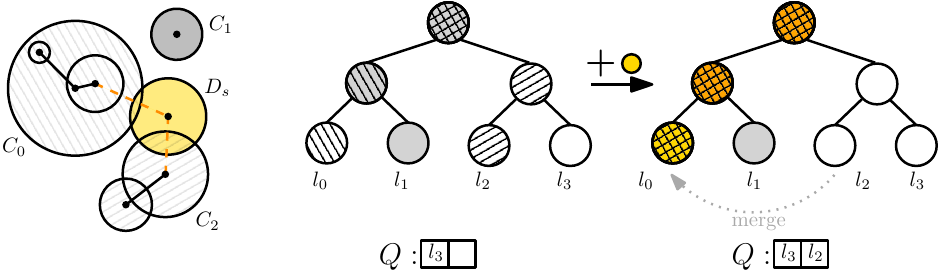}
    \caption{Inserting a site if \(|\mathcal{C}_s| > 1\) ($C_i$ stored in $l_i$): $D_s$ intesects $C_0$ and $C_2$. Since $|C_0|=3$ and $|C_2|=2$, the largest component $C_L$ is $C_0$. Thus, $D_s$ and $C_2$ are merged into $C_0$ and $C_2$ is removed from $l_2$ up to the $\lca$, the root. The empty leave $l_2$ is enqueued.}
    \label{fig:inscleanup}
    \end{figure}
    To show correctness, we again  
    argue that the invariants are maintained. 
    If $|\mathcal{C}_s| = 0$ this follows by \autoref{lm:insertComp}.
    In the other case, we directly insert $s$ 
    into a connected component intersected by $s$ 
    and update all \AWNN structures along the way. 
    As \autoref{lm:find} correctly finds all relevant connected components, 
    and we explicitly move the sites in these components to $C_L$ 
    during clean-up, Invariant~1 is fulfilled.
    In a similar vein, we update all \AWNN structures 
    of sites that move to a new connected component, satisfying Invariant~2.
    Moreover, we keep $Q_\mathcal{C}$ updated by inserting or 
    removing empty leaves when needed during the algorithm.

    To complete the proof, it remains to analyze the running time. 
   In the worst case, where all components are singletons, a 
   component tree that stores $n$ sites has height $O(\log n)$.
    If $|\mathcal{C}_s| = 0$ the running time for finding 
    $\mathcal{C}_s$ is $O(\log^2 n)$ by \autoref{lm:find}.
    The insertion and restructuring is done with \autoref{lm:insertComp}, 
    yielding an expected amortized time of $O(\log^3 n)$.
    In the case $|\mathcal{C}_s| = 1$, with $h = O(\log n)$ 
    a running time of $O(\log^3 n)$ for finding $\mathcal{C}_s$ 
    follows by \autoref{lm:find}.
    Following similar arguments to the case $|\mathcal{C}_s| = 0$, 
    the time needed for the insertion and restructuring is 
    expected amortized $O(\log^3 n)$.
    
    Finally, we consider the case $|\mathcal{C}_s| > 1$.
    By \autoref{lm:find}, finding $\mathcal{C}_s$ 
    takes worst case time $O(|\mathcal{C}_s|\cdot\log^3 n)$.
    Then the insertion of $s$ can be done in expected amortized time 
    $O(\log^3 n)$, as in the cases above.
    It remains to analyze the running time of the clean-up step. 
    We know that the first common ancestor might be the root of 
    $T_\mathcal{C}$.
    Hence, in the worst case, we have to perform 
    $\sum_{C_i \in (\mathcal{C}_s \setminus \{C_L\})} |C_i|\cdot O(\log n)$ 
    insertions and deletions for a single clean-up step.
    As the time for the deletions in the \AWNN structures dominates, 
    this is expected amortized 
    $\sum_{C_i \in \mathcal{C}_s \setminus \{C_L\}} O(|C_i|\cdot \log^5 n)$ 
    worst case time. 
    Note that since $|C_i| \geq 1$ and we have to insert $s$, 
    the running time of \autoref{lm:find} is dominated by the clean-up step.

    The overall time spent on all clean-up steps over all insertions
    is then upper bounded by 
    $\sum_{s\in S} \sum_{C_i\in (\mathcal{C}_s\setminus \{C_L\})} 
    O(|C_i|\cdot \log^5 n)$.
    Observe, that during the lifetime of the component tree, 
    each disk can only be merged into $O(\log n)$ connected components. 
    Thus, we have 
    \[
    \sum_{s\in S} \sum_{C_i\in \mathcal{C}_s\setminus \{C_L\}} |C_i| = 
    O(n\log n), 
    \]
    and the overall expected time spent on clean-up steps is $O(n\log^6 n)$.
    As the case $|\mathcal{C}_s| > 1$ 
    turned out to be the most complex case, the overall running time follows. 
\end{proof}

\begin{theorem}\label{thm:ins_unbounded}
There is an incremental data structure for 
connectivity queries in disk graphs with $O(\alpha(n))$ 
amortized query time and $O(\log^6 n)$ expected amortized update time.
\end{theorem}
\begin{proof}
We use a component tree as described above to maintain the 
connected components. Additionally, we maintain a disjoint set 
data structure $\mathcal{H}$, where each connected component forms a set, see \autoref{fig:ins_unbounded}. 
Queries are performed directly in $\mathcal{H}$ in $O(\alpha(n))$ 
amortized time.

When inserting an isolated component during the update, 
this component is added to $\mathcal{H}$.
When merging several connected components 
in the clean-up step, this change is reflected in $\mathcal{H}$ 
by suitable union operations.
The time for updates in the component tree dominates the updates in 
$\mathcal{H}$, leading to an expected amortized update time of $O(\log^6 n)$.
\end{proof}

\section{Conclusion}

We introduced a data structure that solves the incremental 
connectivity problem in general disk graphs with $O(\alpha(n))$ 
amortized query and $O(\log ^6 n)$ amortized expected 
update time.
The question of finding efficient fully-dynamic 
data structures for both the general and the bounded case remains open.

\bibliography{connectivity}

\begin{thebibliography}{1}

\bibitem{cormen_introduction_2009}
Thomas~H. Cormen, Charles~E. Leiserson, Ronald~L. Rivest, and Clifford Stein.
\newblock {\em Introduction to Algorithms, Third Edition}.
\newblock {The MIT Press}, 3rd edition, 2009.

\bibitem{kaplan_dynamic_2022a}
Haim Kaplan, Alexander Kauer, Katharina Klost, Kristin Knorr, Wolfgang Mulzer,
  Liam Roditty, and Paul Seiferth.
\newblock Dynamic connectivity in disk graphs.
\newblock In {\em 38th International Symposium on Computational Geometry (SoCG
  2022)}, pages 49:1--49:17, 2022.
\newblock \href {https://doi.org/10.4230/LIPIcs.SoCG.2022.49}
  {\path{doi:10.4230/LIPIcs.SoCG.2022.49}}.

\bibitem{kaplan_dynamic_2020}
Haim Kaplan, Wolfgang Mulzer, Liam Roditty, Paul Seiferth, and Micha Sharir.
\newblock Dynamic planar {V}oronoi diagrams for general distance functions and
  their algorithmic applications.
\newblock {\em Discrete Comput. Geom.}, 64(3):838--904, 2020.
\newblock \href {https://doi.org/10.1007/s00454-020-00243-7}
  {\path{doi:10.1007/s00454-020-00243-7}}.

\bibitem{Liu20}
Chih-Hung Liu.
\newblock Nearly optimal planar $k$ nearest neighbors queries under general
  distance functions.
\newblock {\em SIAM Journal on Computing}, 51(3):723--765, 2022.
\newblock \href {https://doi.org/10.1137/20m1388371}
  {\path{doi:10.1137/20m1388371}}.

\bibitem{reif_topological_1987}
John~H. Reif.
\newblock A topological approach to dynamic graph connectivity.
\newblock {\em Information Processing Letters}, 25(1):65--70, 1987.
\newblock \href {https://doi.org/10.1016/0020-0190(87)90095-0}
  {\path{doi:10.1016/0020-0190(87)90095-0}}.

\end{thebibliography}
\end{document}